\documentclass{llncs}

\usepackage{makeidx}  
\usepackage{color}
\usepackage{graphicx}
\usepackage[nounderscore]{syntax}
\usepackage{courier}
\usepackage{hyperref}
\usepackage{amsmath}
\usepackage{float}
\usepackage{url}
\usepackage{bbold}
\usepackage{dsfont}
\usepackage{cite}
\usepackage[labelfont=bf, labelsep=period]{caption}

\newfloat{requirement}{h}{req}
\floatname{requirement}{Req.}

\begin{document}
%
%
\pagestyle{plain}  
%

\title{On Implementing Real-time Specification Patterns Using Observers\thanks{This work was sponsored by DARPA/AFRL Contract FA8750-12-9-0179, AFRL Contract FA8750-16-C-0018, and NASA Contract NNA13AA21C}}
\author{John D. Backes\inst{1}, Michael W. Whalen\inst{2}, Andrew Gacek\inst{1}, John Komp\inst{2}}
\institute{Rockwell Collins, Bloomington MN 55438 \and University of Minnesota, Minneapolis MN 55455}
\maketitle
\begin{abstract}
English language requirements are often used to specify the behavior of complex cyber-physical systems.  The process of transforming these requirements to a formal specification language is often challenging, especially if the specification language does not contain constructs analogous to those used in the original requirements. For example, requirements often contain real-time constraints, but many specification languages for model checkers have discrete time semantics.  Work in specification patterns helps to bridge these gaps, allowing straightforward expression of common requirements patterns in formal languages.  In this work we demonstrate how we support real-time specification patterns in the Assume Guarantee Reasoning Environment (AGREE) using observers.  We demonstrate that there are subtle challenges, not mentioned in previous literature, to express real-time patterns accurately using observers.  We then demonstrate that these patterns are sufficient to model real-time requirements for a real-world avionics system.
\end{abstract}

\section{Introduction}

Natural language requirements specifications are often used to
prescribe the behavior of complex cyber-physical systems.
Regrettably, such specifications can be incomplete, inconsistent, or
ambiguous. For these reasons, researchers have long advocated the use
of formal languages, such as temporal logics to describe requirements.
Unfortunately, the process of formalizing natural language
requirements using formal specification languages is often
challenging, especially if the specification language does not contain
constructs analogous to those used in the original requirements.

Specification patterns~\cite{dwyer1999patterns,konrad2005real} are an approach to ease the construction of formal specifications from natural language requirements.  These patterns describe how common reasoning patterns in English language requirements can be represented in (sometimes complex) formulas in a variety of formalisms.  Following the seminal work of Dwyer~\cite{dwyer1999patterns} for discrete time specification patterns, a variety of real-time specification pattern taxonomies have been developed~\cite{konrad2005real,gruhn2006patterns,bellini2009expressing,Reinkemeir12,etzien2013contracts}.  An example of a timed specification pattern expressible in each is: ``Globally, it is always the case that if $P$ holds, then $S$ holds between {\em low} and {\em high} time unit(s).''

In most of this work, the specification patterns are mapped to real-time temporal logics, such as TCTL~\cite{alur1991techniques}, MTL~\cite{koymans1990specifying},  RTGIL~\cite{moser1997graphical}, and TILCO-X~\cite{bellini2009expressing}.
As an alternative, researchers have investigated using {\em observers} to capture real-time specification patterns.  Observers are code/model fragments written in the modeling or implementation language to be verified, such as timed automata, timed Petri nets, source code, and Simulink, among others.  For example, Gruhn~\cite{gruhn2006patterns} and Abid~\cite{abid2012real} describe real-time specifications as state machines in timed automata and timed Petri nets, respectively.  A benefit of this approach is that rather than checking complex timed temporal logic properties (which can be very expensive and may not be supported by a wide variety of analysis tools), it is possible to check simpler properties over the observer.

Despite this benefit, capturing real-time specification patterns with
observers can be challenging, especially in the presence of
overlapping ``trigger events.'' That is, if $P$ occurs multiple
times before {\em low} time units have elapsed in the example above.
For example, most of the observers in Abid~\cite{abid2012real}
explicitly are not defined for `global' scopes, and Gruhn, while
stating that global properties are supported, only checks a pattern
for the first occurrence of the triggering event in an infinite trace.

In this work, we examine the use of observers and invariant properties
to capture specification patterns that can involve overlapping
triggering events. We use the Lustre specification
language~\cite{lustre} to describe {\em synchronous observers}
involving a real-valued time input to represent the current system
clock\footnote{Although our formalisms are expressed as Lustre
  specifications, the concepts and proofs presented in this paper are
  applicable to many other popular model checking specification
  languages.}. We describe the conditions under which we can use
observers to faithfully represent the semantics of patterns, for both
positive instances of patterns {\em and negations of patterns}. We
call the former use {\em properties} and the latter use {\em
  constraints}.
  
The reason that we consider negations of patterns is that our overall goal is to use real-time specification patterns in the service of assume/guarantee compositional reasoning.  In recent efforts~\cite{Backes15,Murugesan14}, we have used the AGREE tool suite~\cite{CoferNFM2012} for reasoning about discrete time behavioral properties of complex models described in the Architectural Analysis and Design Language~\cite{aadl}\footnote{AGREE is available at: http://loonwerks.com}.  Through adding support for Requirements Specification Language (RSL) patterns~\cite{CESAR} and calendar automata~\cite{Dutertre04,Pike05,Vanwyk07slap}, it becomes possible to lift our analysis to real-time systems.  In AGREE, we prove implicative properties: given that subcomponents satisfy their contracts, then a system should satisfy its contract.  This means that the RSL patterns for subsystems are used under a negation.  We describe the use of these patterns in AGREE and demonstrate their use on a real avionics system.  Thus, the contributions of this work are as follows:

\begin{itemize}
\item We demonstrate a method for translating RSL Patterns into Lustre observers and system invariants.
\item We prove that it is possible to efficiently capture patterns involving arbitrary overlapping intervals in Lustre using non-determinism.
\item We argue that there is no method to efficiently encode a transition system in Lustre that implements the exact semantics of all of the RSL patterns when considering their negation.
\item We demonstrate how to encode these patterns as Lustre constraints for \textit{practical} systems.
\item We discuss the use of these patterns to model a real-world avionics system.
\end{itemize}


\section{Definitions}\label{sec:definitions}

AGREE proves properties of architectural models compositionally by
proving a series of lemmas about components at different levels in the
model's hierarchy. A description of how these proofs are constructed
is provided in~\cite{CoferNFM2012,Backes15} and a proof sketch of
correctness of these rules is described
in~\cite{AgreeUsersGuide,CoferNFM2012}. For the purpose of this work,
it is not important that the reader has an understanding of how these
proofs are constructed. The AGREE tool translates AADL models
annotated with component assumptions, guarantees, and assertions into
Lustre programs. Our explanations and formalizations in this paper are
described by these target Lustre specifications. Most other SMT-based
model checkers use a specification language that has similar
expressivity as Lustre; the techniques we present in this paper can be
applied generally to other model checking specification languages.

A Lustre program $\mathcal{M} = (V, T, P)$ can be thought of as a
finite collection of named variables $V$, a transition relation $T$,
and a finite collection of properties $P$. Each named variable is of
type $bool$, $integer$, or $real$. The transition relation is a
Boolean constraint over these variables and theory constants; the
value of these variables represents the program's current
\textit{state}, and the transition relation constrains how the state
changes. Each property $p \in P$ is also a Boolean constraint over the
variables and theory constants. We sometimes refer to a Lustre program
as a model, specification, or transition system. The AGREE constraints
specified via assumptions, assertions, or guarantees in an AADL model
are translated to either constraints in the transition relation or
properties of the Lustre program.

The expression for $T$ contains common arithmetic and logical
operations ($+$, $-$, $*$, $\div$, $\vee$, $\wedge$, $\Rightarrow$,
$\neg$, $=$) as well as the ``if-then-else'' expression ($ite$) and
two temporal operations: $\rightarrow$ and $pre$. The $\rightarrow$
operation evaluates to its left hand side value when the program is in
its initial state. Otherwise it evaluates to its right hand side
value. For example, the expression: $true \rightarrow false$ is $true$
in the initial state and $false$ otherwise. The $pre$ operation takes
a single expression as an argument and returns the value of this
expression in the previous state of the transition system. For
example, the expression: $x = (0 \rightarrow pre (x) + 1)$ constrains
the current value of variable $x$ to be $0$ in the initial state
otherwise it is the value of $x$ in the previous state incremented by
$1$.

In the model's initial state the value of the $pre$ operation on any
expression is undefined. Every occurrence of a $pre$ operator must be
in a subexpression of the right hand side of the $\rightarrow$
operator. The $pre$ operation can be performed on expressions
containing other $pre$ operators, but there must be $\rightarrow$
operations between each occurrence of a $pre$ operation. For example,
the expression: $true \rightarrow pre(pre(x))$ is not well-formed, but
the expression: $true \rightarrow pre(x \rightarrow pre(x))$ is
well-formed.

A Lustre program models a state transition system. The current values
of the program's variables are constrained by values of the program's
variables in the previous state. In order to model timed systems, we
introduce a real-valued variable $t$ which represents how much time
has elapsed during the previous transitions of the system. We adopt a
similar model as \textit{timeout automata} as described
in~\cite{Dutertre04}. The system that is modeled has a collection of
\textit{timeouts} associated with the time of each ``interesting
event'' that will occur in the system. The current value of $t$ is
assigned to the least timeout of the system greater than the previous
elapsed time. Specifically, $t$ has the following constraint:
\begin{equation}\label{fml:calendar}
t = 0 \rightarrow pre(t) + min\_pos(t_1 - pre(t), \ldots, t_n - pre(t))
\end{equation}
where $t_1, \ldots, t_n$ are variables representing the timeout values
of the system. The function $min\_pos$ returns the value of its
minimum positive argument. We constrain all the timeouts of the system
to be positive. A timeout may also be assigned to positive infinity
($\infty$)\footnote{In practice, we allow a timeout to be a negative
  number to represent infinity. This maintains the correct semantics
  for the constraint for $t$ in Formula~\ref{fml:calendar}.}. There
should always be a timeout that is greater than the current time (and
less than $\infty$). If this is true, then the invariant $true
\rightarrow t > pre(t)$ holds for the model, i.e., time always
progresses.

A sequence of states is called a \textit{trace}. A trace is said to be
\textit{admissible} (w.r.t. a Lustre model or transition relation) if
each each state and its successor satisfy the transition relation. We
adopt the common notation $(\sigma,\tau)$ to represent a trace of a
timed system where $\sigma$ is a sequence of states ($\sigma =
\sigma_1\sigma_2\sigma_3\ldots$) and $\tau$ is a sequence of time
values ($\tau = \tau_1\tau_2\tau_3\ldots$) such that $\forall i :
\tau_i < \tau_{i+1}$. In some literature, state transitions may take
place without any time progress (i.e., $\forall i : \tau_i \leq
\tau_{i+1}$). We do not allow these transitions as it dramatically
increases the complexity of a model's Lustre encoding.

A Lustre program implicitly describes a set of admissible traces. Each
state $\sigma_n$ in the sequence represents the value of the variables
$V$ in state $n$. Each time value $\tau_n$ represents the value of the
time variable $t$ in state $n$. We use the notation $\sigma_n \models
e$, where $e$ is Lustre expression over the variables $V$ and theory
constants, if the expression $e$ is satisfied in the state $\sigma_n$.
Similarly, we use $\sigma_n \not\models e$ when $e$ is not satisfied
in the state $\sigma_n$. A property $p$ is true (or invariant) in a
model if and only if for every admissible trace $\forall n : \sigma_n
\models p$. For the purposes of this work, we only consider models
that do not admit so-called ``Zeno traces''~\cite{Gomez2007}. A trace
$(\sigma, \tau)$ is a Zeno trace if and only if $\exists v \forall i :
\tau_i < v$, i.e., time never progresses beyond a fixed point.


\section{Implementing RSL Patterns}\label{sec:patterns}

\subsection{Formalizing RSL Patterns Semantics}

For this work, we chose to target the natural language patterns proposed in the CESAR project because they are representative of many types of natural language requirements~\cite{CESAR}. These patterns are divided into a number of categories.  The categories of interest for this work are the \textit{functional patterns} and the \textit{timing patterns}. Some examples of the functional patterns are:
\begin{enumerate}
\item \textbf{Whenever} event \textbf{occurs} event \textbf{occurs} \textbf{during} interval
\item \textbf{Whenever} event \textbf{occurs} condition \textbf{holds} \textbf{during} interval
\item \textbf{When} condition \textbf{holds} \textbf{during} interval event \textbf{occurs} \textbf{during} interval
\item \textbf{Always} condition
\end{enumerate}
Some examples of timing patterns are:
\begin{enumerate}
\item Event \textbf{occurs each} period \textbf{[with jitter} jitter\textbf{]}
\item Event \textbf{occurs sporadic with IAT} interarrivaltime \textbf{[and jitter} jitter\textbf{]}
\end{enumerate}

Generally speaking, the timing patterns are used to constrain how often a system is required to respond to events. For instance, a component that listens to messages on a shared bus might assume that new messages arrive at most every 50ms. The second timing pattern listed above would be ideal to express this assumption. In AGREE, this requirement may appear as a system assumption using the pattern shown in Figure~\ref{fig:sporadic}.

\begin{figure}
\centering
\fbox{
\begin{minipage}{0.55\textwidth}
\textit{new\_message} \textbf{occurs sporadic with IAT} 50.0
\end{minipage}
}
\caption{An instance of a timing pattern to represent how frequently a message arrives on a shared bus.}
\label{fig:sporadic}
\end{figure}

The functional patterns can be used to describe how the system's state changes in response to external stimuli. Continuing with the previous example, suppose that the bus connected component performs some computation whenever a new message arrives. The functional patterns can be used to describe when a thread is scheduled to process this message and how long the thread takes to complete its computation. The intervals in these patterns have a specified lower and upper bound, and they may be open or closed.  The time specified by the lower and upper bound corresponds to the time that progresses since the triggering event occurs. Both the lower and upper bounds must be positive real numbers, and the upper bound must be greater than or equal to the lower bound. An AGREE user may specify the instances of patterns shown in Figure~\ref{fig:event-pattern} as properties she would like to prove about this system.  For the purposes of demonstration we assume that the thread should take 10ms to 20ms to execute.

\begin{figure}
\centering
\fbox{
\begin{minipage}{0.85\textwidth}\textbf{Always} \textit{new\_message = thread\_start} \\
\textbf{Whenever} \textit{thread\_start} \textbf{occurs} \textit{thread\_stop} \textbf{occurs during} \textbf{[}10.0\textbf{,} 20.0\textbf{]}
\end{minipage}
}
\caption{Two instances of a functional patterns used to describe when a thread begins executing, and how long it takes to execute.}
\label{fig:event-pattern}
\end{figure}

\begin{figure}[t]
\centering
\includegraphics[scale=1]{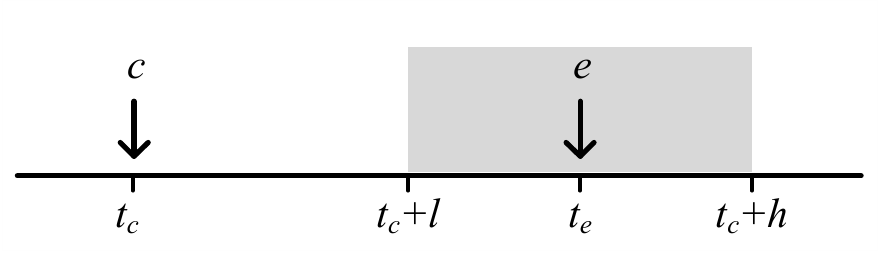}\\
\textbf{whenever} $c$ \textbf{occurs} $e$ \textbf{occurs during} \textbf{[}$l$\textbf{,} $h$\textbf{]}
\caption{A graphical representation for the RSL pattern}\label{fig:p1}
\end{figure}

Figure~\ref{fig:p1} shows a graphical representation of the first
functional pattern listed at the beginning of this section. The
variable $t_{c}$ represents the time that event $c$ occurs. Similarly,
the variable $t_{e}$ represents the time that event $e$ occurs. The
formal semantics for many of the RSL patterns are described
in~\cite{Reinkemeir12}. The semantics for the pattern described in
Figure~\ref{fig:p1} are represented by the set of admissible traces
$\mathcal{L}_{patt}$ described below.
\begin{equation*}
\mathcal{L}_{patt} = \{(\sigma,\tau)~|~\forall i \exists j :
  \sigma_i \models c \Rightarrow (j > i) \wedge (\tau_i + l \leq \tau_j \leq \tau_i + h) \wedge (\sigma_j \models e)\}
\end{equation*}
The remainder of this section discusses how the pattern in
Figure~\ref{fig:p1} can be translated into either a Lustre property or
a constraint on the admissible traces of a transition system
described by Lustre. Although we discuss only this pattern, the
techniques that we present can be applied generally to all except one of the functional and timing RSL patterns\footnote{The single pattern that cannot be implemented requires an independent event to occur for each of an unbounded number of causes. There are 12 functional and timing RSL patterns in total.}.


\subsection{Implementing RSL Patterns as Lustre Properties}
\label{sec:rsl-to-lustre}

One can determine if a transition system described in Lustre admits only traces in $\mathcal{L}_{patt}$ by adding additional constraints over fresh variables (variables that are not already present in the program) to the model. This commonly used technique is referred to as adding an \textit{observer} to the model. These constraints are over fresh variables: $run, timer, rec_{c}$ and $pass$; they are shown in Figure~\ref{fig:lustre-prop}.  The constraints only restrict the values of the fresh variables, therefore they do not restrict the traces admissible by the transition relation.
\begin{figure}
\centering
\fbox{
\begin{minipage}{.92\textwidth}
\begin{enumerate}
\item $run = (rec_{c} \rightarrow ite( pre(run) \wedge e \wedge l \leq timer \leq h,\\
\phantom{xxxxxxxxxxxxxxx} false, \\
\phantom{xxxxxxxxxxxxxxx} ite( rec_{c}, true, pre(run))))$
\item $timer = (0 \rightarrow ite( pre(run), pre(timer) + (t - pre(t)), 0))$
\item $rec_{c} \Rightarrow c$
\item $pass = (timer \leq h)$
\end{enumerate}
\end{minipage}
}
\caption{The constraints added to a transition relation to verify if only the traces of $\mathcal{L}_{patt}$ are admissible. The transition relation only admits traces of $\mathcal{L}_{patt}$ if and only if the variable $pass$ is invariant.}\label{fig:lustre-prop}
\end{figure}

The intuition behind these constraints is that one can record how much
time progresses since an occurrence of $c$. This time is recorded in
the $timer$ variable. The value of the timer variable only increases
if the previous value of the $run$ variable is true. The $run$
variable is true if an occurrence of $c$ is \textit{recorded} and no
occurrence of $e$ happens until after the timer counts to at least
$l$. The variable $rec_{c}$ non-deterministically records an
occurrence of $c$. If the transition system admits a trace outside of
$\mathcal{L}_{patt}$, then the $rec_{c}$ variable can \textit{choose}
to record only an event that violates the conditions of
$\mathcal{L}_{patt}$. In this case the $pass$ variable will become
false in some state.

\begin{theorem}\label{thm:properties}
Let $\mathcal{L}_{M}$ represent the admissible traces of a transition system containing the constraints of Figure~\ref{fig:lustre-prop}. The transition system admits only traces in $\mathcal{L}_{patt}$ if and only if the property $pass$ is invariant. Formally: $(\mathcal{L}_{M} \subseteq \mathcal{L}_{patt}) \Leftrightarrow (\forall \sigma, \tau, i : (\sigma, \tau) \in \mathcal{L}_{M} \Rightarrow \sigma_i \models pass)$
\end{theorem}

\begin{proof}
First we show that if $pass$ is invariant for a trace of the transition relation, then that trace is in $\mathcal{L}_{patt}$.

\begin{lemma}\label{proplem1}
 $(\forall \sigma, \tau, i : (\sigma, \tau) \in \mathcal{L}_{M} \Rightarrow \sigma_i \models pass) \Rightarrow (\mathcal{L}_{M} \subseteq \mathcal{L}_{patt})$.
\end{lemma}
\begin{proof}
Towards contradiction, assume $\mathcal{L}_{M} \not\subseteq
\mathcal{L}_{patt}$. Let $(\sigma, \tau)$ be a trace in
$\mathcal{L}_M$ but not in $\mathcal{L}_{patt}$. Since $(\sigma, \tau)
\notin \mathcal{L}_{patt}$, by definition there exists $i$ such that
$\sigma_i \models c$ and
\begin{equation}\label{fml:not-e}
\forall j : (j > i) \wedge \tau_i+l \leq
\tau_j \leq \tau_i+h \Rightarrow \sigma_j \not\models e.
\end{equation}
Without loss of generality, we can assume that this is the only time
when $c$ is recorded. That is, $\sigma_i \models rec_{c}$ and $\forall
k : k \neq i \Rightarrow \sigma_k \not\models rec_{c}$. From
constraint 1 in Figure~\ref{fig:lustre-prop} we have
\begin{equation*}
\forall j : ((j < i) \Rightarrow \sigma_j \not\models run)
\wedge ((\tau_i \leq \tau_j < \tau_i+l) \Rightarrow \sigma_j \models run)
\end{equation*}
This can actually be strengthened more. From Formula~\ref{fml:not-e}
the event $e$ does not occur between $\tau_i + l$ and $\tau_i + h$. So
the variable $run$ will become invariant after $\tau_i$.
\begin{equation*}
\forall j : ((j < i) \Rightarrow \sigma_j \not\models run) \wedge
(\tau_i \leq \tau_j) \Rightarrow \sigma_j \models run)
\end{equation*}
From this and constraint 2 in Figure~\ref{fig:lustre-prop}, we have
\begin{equation*}
\forall j : (j \leq i) \Rightarrow \sigma_j \models timer = 0
\end{equation*}
and
\begin{equation*}
\forall j : (\tau_i < \tau_j) \Rightarrow (\sigma_j \models timer =
(pre(timer) + (\tau_j - \tau_{j-1})))
\end{equation*}
From this and the invariant $\forall i : \tau_{i+1} > \tau_i$, we have
\begin{equation*}
\begin{split}
&\forall j : (\tau_i < \tau_j) \Rightarrow (\sigma_j \models timer > pre(timer))
\end{split}
\end{equation*}
Therefore since the value of timer is zero before $\tau_i$ and always
increasing after $\tau_i$, and since we only consider non-Zeno traces
($\forall v \exists i : v < \tau_i$), eventually $timer > h$ and so
$pass$ becomes false. This contradicts the assumption $(\forall
\sigma, \tau, i : (\sigma, \tau) \in \mathcal{L}_{M} \Rightarrow
\sigma_i \models pass)$. Therefore $\mathcal{L}_{M} \subseteq
\mathcal{L}_{patt}$. \qed
\end{proof}

Next we show if a trace of $\mathcal{L}_{M}$ is in
$\mathcal{L}_{patt}$, then $pass$ is invariant for this trace.

\begin{lemma}\label{proplem2}
 $(\mathcal{L}_{M} \subseteq \mathcal{L}_{patt}) \Rightarrow (\forall
  \sigma, \tau, i : (\sigma, \tau) \in \mathcal{L}_{M} \Rightarrow
  \sigma_i \models pass)$
\end{lemma}

\begin{proof}
Towards contradiction, assume that there exists a trace of
$\mathcal{L}_{M}$ for which $pass$ is not invariant. This means that
for some state $\sigma_j \models timer > h$. For this to be true, the
timer must be running continuously since it started with some recorded
occurrence of $c$. That is there exists $i$ such that $\sigma_i\models
timer = 0$, $\sigma_i\models rec_c$, $\sigma_i\models c$, $\forall k :
i \leq k < j \Rightarrow \sigma_k\models run$, and $\tau_j - \tau_i
> h$. Thus $\forall k : i \leq k \leq j \Rightarrow \sigma_k\models
timer = \tau_k - \tau_i$. By the definition of $\mathcal{L}_{patt}$ we
have a $k$ such that $\tau_i + l \leq \tau_k \leq \tau_i + h$ and
$\sigma_k \models e$. This means $l \leq \tau_k - \tau_i \leq h$ and
so $\sigma_k\models l \leq timer \leq h$. Therefore $\sigma_k
\not\models run$. We also have $\tau_k \leq \tau_i + h < \tau_j$ so
that $k < j$. Thus from $\forall k : i \leq k < j \Rightarrow
\sigma_k\models run$ we have $\sigma_k\models run$ which is a
contradiction. Therefore, $pass$ is invariant. \qed
\end{proof}

From Lemmas~\ref{proplem1} and~\ref{proplem2} we have $(\mathcal{L}_{M}
\subseteq \mathcal{L}_{patt}) \Leftrightarrow (\forall \sigma, \tau, i
: (\sigma, \tau) \in \mathcal{L}_{M} \Rightarrow \sigma_i \models
pass)$. \qed
\end{proof}


\subsection{Implementing RSL Patterns as Lustre Constraints}

As we demonstrated with Figure~\ref{fig:lustre-prop}, one can specify
a Lustre property that verifies whether or not some transition system
only admits traces of $\mathcal{L}_{patt}$. However, it is
surprisingly non-trivial to actually implement a transition system
that admits \textit{exactly} the traces of $\mathcal{L}_{patt}$.
Naively, one could attempt to add the constraints of
Figure~\ref{fig:lustre-prop} to a transition system and then assert
that $pass$ is invariant. However, this transition system will admit
all traces where every occurrence of $c$ is never recorded ($\forall
\sigma_i : \sigma_i \not\models rec_{c}$). Clearly some of these
traces would not be in $\mathcal{L}_{patt}$.

We conjecture that given the Lustre expression language described in
Section~\ref{sec:definitions} it is not possible to model a transition
system that admits only and all of the traces of $\mathcal{L}_{patt}$.
The intuition behind this claim is that Lustre specifications contain
a fixed number of state variables, and variables have non-recursive
types. Thus a Lustre specification only has a finite amount of memory
(though it can, for example, have arbitrary sized integers). If a
Lustre specification has $n$ variables we can always consider a trace
in $\mathcal{L}_{patt}$ where event $c$ occurs more than $n$ times in
a tiny interval. In order for the pattern to hold true, the Lustre
specification must constrain itself so that at least one occurrence of
$e$ occurs precisely between $t_{c}+l$ and $t_{c}+h$ after each event
$c$. This requires ``more memory'' than the Lustre specification has
available.

Rather than model the exact semantics of this pattern, we choose to
take a more pragmatic approach. We model a strengthened version of
Figure~\ref{fig:p1} which does not allow overlapping instances of the
pattern. That is, after an event $c$ there can be no more occurrences
of $c$ until the corresponding occurrence of $e$. We do this by proving
that $c$ cannot occur frequently enough to cause an overlapping
occurrence of the pattern. Then if we constrain the system based on a
simple non-overlapping check of the pattern, the resulting system is
the same as if we had constrained it using the full pattern. This
simple non-overlapping check and the property limiting the frequency
of $c$ are both easily expressed in Lustre since they only look back
at the most recent occurrence of $c$. Moreover, they can both be used
freely in positive and negative contexts. Formally, the property we
prove is $\mathcal{L}_{prop}$ and the constraints we make are
$\mathcal{L}_{cons}$:
\begin{equation*}
\begin{split}
\mathcal{L}_{prop} &= \{(\sigma,\tau)~|~\forall i :
  \sigma_i \models c \Rightarrow \forall j : (j > i) \wedge
  (\tau_j \leq \tau_i + h) \wedge \sigma_j \models c
  \Rightarrow\\
  &\quad \exists k \in (i, j] : \tau_i + l \leq \tau_k \wedge \sigma_k \models e\}
\end{split}
\end{equation*}
\begin{equation*}
\begin{split}
\mathcal{L}_{cons} &= \{(\sigma,\tau)~|~\forall i:
  \sigma_i \models c \Rightarrow \exists j : (j > i)\wedge~\\
  &\quad [(\tau_i + l \leq \tau_j \leq \tau_i + h \wedge \sigma_j \models e) \vee (\tau_j \leq \tau_i + h \wedge \sigma_j \models c)] \}
\end{split}
\end{equation*}
The correctness of $\mathcal{L}_{prop}$ and $\mathcal{L}_{cons}$ are
captured by the following theorem.

\begin{theorem}\label{theorem2}
Let $M$ be a transition system and $\mathcal{L}_{M}$ its corresponding
set of admissible traces. Suppose $\mathcal{L}_M \subseteq
\mathcal{L}_{prop}$. Then $\mathcal{L}_{cons}$ and
$\mathcal{L}_{patt}$ are equivalent restrictions on $\mathcal{L}_M$,
that is $\mathcal{L}_{M} \cap \mathcal{L}_{cons} = \mathcal{L}_{M}
\cap \mathcal{L}_{patt}$.
\end{theorem}

\begin{proof}
We prove the theorem by showing that the subset relationship between
$\mathcal{L}_{M} \cap \mathcal{L}_{cons}$ and $\mathcal{L}_{M} \cap
\mathcal{L}_{patt}$ holds in both directions.

\begin{lemma}\label{lemma3}
$\mathcal{L}_{M} \cap \mathcal{L}_{patt} \subseteq \mathcal{L}_{M}
\cap \mathcal{L}_{cons}$
\end{lemma}

\begin{proof}
From the definitions of $\mathcal{L}_{patt}$ and $\mathcal{L}_{cons}$
it follows directly that $\mathcal{L}_{patt} \subseteq
\mathcal{L}_{cons}$. Therefore $\mathcal{L}_{M} \cap
\mathcal{L}_{patt} \subseteq \mathcal{L}_{M} \cap \mathcal{L}_{cons}$.
\qed
\end{proof}

\begin{lemma}\label{lemma4}
Suppose $\mathcal{L}_{M} \subseteq \mathcal{L}_{prop}$, then
$\mathcal{L}_{M} \cap \mathcal{L}_{cons} \subseteq \mathcal{L}_{M}
\cap \mathcal{L}_{patt}$
\end{lemma}

\begin{proof}
Suppose towards contradiction that $\mathcal{L}_{M} \cap
\mathcal{L}_{cons} \not\subseteq \mathcal{L}_{M} \cap
\mathcal{L}_{patt}$. Consider a trace $(\sigma, \tau) \in
\mathcal{L}_{M} \cap \mathcal{L}_{cons}$ with $(\sigma, \tau) \notin
\mathcal{L}_M \cap \mathcal{L}_{patt}$. Then we have $(\sigma, \tau)
\in \mathcal{L}_{cons}$, $(\sigma, \tau) \in \mathcal{L}_{prop}$, and
$(\sigma, \tau) \notin \mathcal{L}_{patt}$. From the definition of
$\mathcal{L}_{patt}$ we have an $i$ such that $\sigma_i\models c$ and
\begin{equation}\label{fml:patt-not-e}
\forall j : (j > i) \wedge (\tau_i + l \leq \tau_j \leq \tau_i + h)
\Rightarrow \sigma_j\not\models e.
\end{equation}
Then from the definition of $\mathcal{L}_{cons}$ with $\sigma_i\models
c$ we have a $j$ such that $j > i$ and either $(\tau_i + l \leq \tau_j
\leq \tau_i + h \wedge \sigma_j \models e)$ or $(\tau_j \leq \tau_i +
h \wedge \sigma_j \models c)$. The former option directly contradicts
Formula~\ref{fml:patt-not-e}, so we must have $\tau_j \leq \tau_i + h$
and $\sigma_j \models c$. From the definition of $\mathcal{L}_{prop}$
with $\sigma_i\models c$ and our $j$, we have a $k$ in $(i, j]$ such
  that $\tau_i + l \leq \tau_k$ and $\sigma_k\models e$. From $k \leq
  j$ we have $\tau_k \leq \tau_j$ and thus $\tau_i + l \leq \tau_k
  \leq \tau_i + h$. Instantiating Formula~\ref{fml:patt-not-e} with
  $k$ yields $\sigma_k\not\models e$, a contradiction. Therefore
  $\mathcal{L}_{M} \cap \mathcal{L}_{cons} \subseteq
  \mathcal{L}_{M} \cap \mathcal{L}_{patt}$. \qed
\end{proof}

From Lemmas~\ref{lemma3} and~\ref{lemma4} have $\mathcal{L}_{M} \cap
\mathcal{L}_{cons} = \mathcal{L}_{M} \cap \mathcal{L}_{patt}$. \qed
\end{proof}




\begin{example}\label{ex:main}
Suppose we want to model a system of components communicating on a
shared bus. The transition relation for this system must contain
constraints that dictate when threads can start and stop and how
frequently new messages may arrive. First we constrain the event
$new\_message$ from occurring too frequently according to the pattern
instance in Figure~\ref{fig:sporadic}. Let $\mathcal{L}_{nm}$
represent the set of admissible traces for this pattern. This set is
defined explicitly in Formula~\ref{fml:new-message}.
\begin{equation*}\label{fml:new-message}
\begin{split}
\mathcal{L}_{nm} &= \{(\sigma,\tau)~|~\forall i : \sigma_i \models new\_message \Rightarrow \\
  & \neg [\exists j : (j > i) \wedge (\tau_j < \tau_i + 50) \wedge (\sigma_j \models new\_message)]\}
\end{split}
\end{equation*}

Suppose we wish to constrain the system to the pattern instances in
Figure~\ref{fig:event-pattern}. The first pattern instance is
represented by the set $\mathcal{L}_{start}$ and the second by 
$\mathcal{L}_{stop}$:
\begin{equation*}\label{fml:start}
\mathcal{L}_{start} = \{(\sigma,\tau)~|~\forall i : \sigma_i \models new\_message \Rightarrow \sigma_i \models thread\_start\}
\end{equation*}
\begin{equation*}\label{fml:stop}
\begin{split}
\mathcal{L}_{stop} &= \{(\sigma,\tau)~|~\forall i \exists j : \sigma_i \models thread\_start \Rightarrow \\
& (j > i) \wedge (\tau_i + l \leq \tau_j \leq \tau_i + h) \wedge (\sigma_j \models thread\_stop)\}
\end{split}
\end{equation*}

Let $\mathcal{L}_{M}$ denote the admissible traces of the transition
system that is being modeled. The goal is to specify the transition
system in Lustre such that $\mathcal{L}_{M} = \mathcal{L}_{nm} \cap
\mathcal{L}_{start} \cap \mathcal{L}_{stop}$. Writing a Lustre
constraint to represent the set of traces $\mathcal{L}_{start}$ is
trivial. The traces that are contained in $\mathcal{L}_{start}$ are
those whose states all satisfy the expression $new\_message =
thread\_stop$. However, as we noted earlier, it is not possible to
develop a set of Lustre constraints that admit only (and all of) the
traces of $\mathcal{L}_{stop}$.

Note that the second pattern in Figure~\ref{fig:event-pattern} is an instance of the pattern described in Figure~\ref{fig:p1}. Therefore we can split the set $\mathcal{L}_{stop}$ into two sets, $\mathcal{L}_{stopc}$ and $\mathcal{L}_{stopp}$:
\begin{equation*}\label{fml:stopc}
\begin{split}
\mathcal{L}_{stopc} &= \{(\sigma,\tau)~|~\forall i:
  \sigma_i \models thread\_start \Rightarrow \exists j : (j > i)\wedge~\\
  &\quad [(\tau_i + l \leq \tau_j \leq \tau_i + h \wedge \sigma_j \models thread\_stop) \vee~ \\
  &\quad (\tau_j \leq \tau_i + h \wedge \sigma_j \models thread\_start)] \}
\end{split}
\end{equation*}
\begin{equation*}\label{fml:stopp}
\begin{split}
\mathcal{L}_{stopp} &= \{(\sigma,\tau)~|~\forall i :
  \sigma_i \models thread\_start \Rightarrow \forall j : (j > i) \wedge~\\
  &\quad (\tau_j \leq \tau_i + h) \wedge \sigma_j \models thread\_start \Rightarrow~ \\
  &\quad \exists k \in (i , j] : \tau_i + l \leq \tau_k \wedge \sigma_k \models thread\_stop\}
\end{split}
\end{equation*}

In this example, the sets of admissible traces representing the patterns happen to have the following relationship:

\begin{equation}\label{exampleeq}
\mathcal{L}_{nm} \cap \mathcal{L}_{start} \subseteq \mathcal{L}_{stopp}
\end{equation}

This is because for every trace in $\mathcal{L}_{nm}$ the event
$new\_message$ only occurs at most every 50ms. Likewise, for each
state of every trace of $\mathcal{L}_{start}$ the variable
$thread\_start$ is true if and only if $new\_message$ is true.
Finally, the set $\mathcal{L}_{stopp}$ contains every trace where
$thread\_start$ occurs at most every 20ms. From
Formula~\ref{exampleeq} and Theorem~\ref{theorem2} we have
$\mathcal{L}_{nm} \cap \mathcal{L}_{start} \cap \mathcal{L}_{stopc} =
\mathcal{L}_{nm} \cap \mathcal{L}_{start} \cap \mathcal{L}_{stop}$.
Thus the system $\mathcal{L}_{nm} \cap \mathcal{L}_{start} \cap
\mathcal{L}_{stopc}$, which we can model in Lustre, is equivalent to a system constrained by
the pattern instances in Figures~\ref{fig:sporadic} and~\ref{fig:event-pattern}.
\end{example}

Example~\ref{ex:main} is meant to demonstrate that, in practical
systems, there is usually some constraint on how frequently events
outside the system may occur. Systems described by the functional RSL
patterns generally have some limitations on how many events they can
respond to within a finite amount of time. The Lustre implementations
of $\mathcal{L}_{cons}$ and $\mathcal{L}_{prop}$ are simpler than
Figure~\ref{fig:lustre-prop}, and their proof of correctness is also
simpler then Theorem~\ref{thm:properties}, though we omit both due to
space limitations.

\section{Application}

We implemented a number of RSL patterns into the AGREE tool. These
patterns were used to reason about the behavior of a real-world
avionics system. Specifically, the patterns were used to model the
logic and scheduling constraints of threads running on a real-time
operating system on an embedded computer on an air vehicle. Each
thread in the system has a single entry point that is dispatched by
some sort of event. The event may be the arrival of data from a bus or
a signal from another thread. When a thread receives an event, the
current state of the thread's inputs are latched. Each thread runs at
the same priority as every other thread (no thread may preempt any
other thread). A thread begins executing after it receives an event
and no other thread is executing.

The patterns in Figures~\ref{fig:sporadic} and~\ref{fig:event-pattern} are actually fairly representative of the constraints used in this model. Figure~\ref{fig:requirements} shows some of the RSL patterns that were used to describe these scheduling constraints. We added an additional tag ``exclusively'' before the second event in the patterns to indicate that the second event occurs only in the specified interval after the first pattern (and never any other time). We found that this was a useful shorthand because one often wants to specify a signal that only occurs under a specified condition and not at any other time.

\begin{figure}
\centering
\fbox{
\begin{minipage}{.88\textwidth}
\textbf{assert} \textit{``thread A runtime''} \textbf{ : whenever} \textit{thread\_A\_start\_running} \textbf{occurs} \\ \phantom{xxx} \textit{thread\_A\_finish}
\textbf{exclusively occurs during [}10.0\textbf{,} 50.0\textbf{];}\\

\textbf{assert} \textit{``thread B runtime''} \textbf{ : whenever} \textit{thread\_B\_start\_running} \textbf{occurs} \\ \phantom{xxx} \textit{thread\_B\_finish}
\textbf{exclusively occurs during [}10.0\textbf{,} 50.0\textbf{];} \\

\textbf{assert} \textit{``thread C runtime''} \textbf{ : whenever} \textit{thread\_C\_start\_running} \textbf{occurs} \\ \phantom{xxx} \textit{thread\_C\_finish}
\textbf{exclusively occurs during [}10.0\textbf{,} 50.0\textbf{];}
\end{minipage}
}
\caption{Assertions about the how the operating system schedules threads}
\label{fig:requirements}
\end{figure}

The results that each thread produces after it finishes executing are
described by an assume-guarantee contract. Generally speaking, the
assumptions restrict the values of inputs that the thread expects to
see. Likewise, the thread's guarantees constrain the values of the
thread's outputs based on it's current state and input values. The
AADL component that contains the threads has assumptions about how
frequently it receives inputs and has guarantees about how quickly it
produces outputs. These assumptions are translated to constraints in
the Lustre transition system, and the guarantees are translated to
properties. Figure~\ref{fig:contract} illustrates one of these
assumptions and guarantees.

The ``eq'' statements in Figure~\ref{fig:contract} are used to
constrain a variable to an expression. They are usually used as a
convenient short hand to make AGREE contracts easier to read. In this
case, the first ``eq'' statement is used to set the variable
$change\_status\_request$ to true if and only if a new message has
arrived and the content of the message is requesting that the vehicle
change its status. Likewise, the second statement is used to record
the last requested change value into the $change\_request$ variable.
The contract assumes that this new message arrives periodically (with
some jitter). The contract guarantees that if a new message arrives
requesting that the vehicle change its status, then the vehicle's
status will be set to the requested value within 500ms. In this
application we assumed that all time units are expressed in
microseconds. This means that the timing constraints expressed in
Figure~\ref{fig:requirements} are also expressed in microseconds.
Other constraints are used to assert that the $vehicle\_status$
variable corresponds to one of the state variables in the component's
threads.

\begin{figure}
\centering
\fbox{
\begin{minipage}{.88\textwidth}
\textbf{eq} change\_status\_event \textbf{: bool =} \\ 
\phantom{xxx} new\_message \textbf{and} message\_content.change\_vehicle\_status\textbf{;} \\

\textbf{eq} change\_request \textbf{: bool =} \\ 
 \phantom{xxx} \textbf{ite(}change\_status\_event\textbf{,} \\
 \phantom{xxx} \phantom{xxx} message\_content.status\textbf{,} \\
 \phantom{xxx} \phantom{xxx} \textbf{false $\rightarrow$ pre(}change\_request\textbf{));} \\

\textbf{assume} \textit{``periodic messages''} \textbf{:} new\_message \textbf{occurs} \\
\phantom{xxx} \textbf{each} 10000.0 \textbf{with jitter} 50.0\textbf{;}\\

\textbf{guarantee} \textit{``new message can change vehicle status''} \textbf{:} \\ 
\phantom{xxx} \textbf{whenever} change\_status\_event \textbf{occurs} \\ 
\phantom{xxx} \phantom{xxx} vehicle\_status = change\_request \textbf{during [}0.0\textbf{,} 500.0\textbf{];}

\end{minipage}
}
\caption{Assumptions and guarantees about the component containing the threads.}
\label{fig:contract}
\end{figure}

The guarantee of this component is invariant if and only if the threads in the component's implementation are scheduled in such a way that whenever a new message arrives its content is parsed and sent to the correct threads to be processed in a timely manner.  The logic expressed in the contract of each thread determines how the content of this message is transmitted to other threads in the system. 

\subsection{Results}

We had three properties of interest for the vehicle. These properties were related to timing, schedulability, and behavior of the system's threads.  We ran the translated Lustre file, which contained about 1000 lines, from the AADL/AGREE model on the latest version of JKind on a Linux machine with an Intel(R) Xeon(R) E5-1650 CPU running at 3.50GHz.  JKind uses k-induction, property directed reachability, and invariant generation engines to prove properties of Lustre models. In the case of this experiment, it took about 8 hours to prove all three properties. One of the properties was proved via k-induction, the other two were proved by the property directed reachability engine.

JKind allows users to export the lemmas used to prove a property. These lemmas can be exported and used again in order to speed up solving for similar models and properties. We found that when these lemmas were used again to prove the properties a second time all of the properties were proved in less than 10 seconds.  This seems to indicate that the properties are not particularly \textit{deep}. That is to say, to prove the properties via k-induction, the inductive step does not need to unroll over many steps.  We are currently exploring techniques for lemma discovery for properties specified with RSL patterns.

\section{Related Work}


Our work focuses on the real-time patterns in the Requirements Specification Language (RSL)~\cite{Reinkemeir12} that was created as part of the CESAR project~\cite{CESAR}.  This language was an extension and modularization of the Contract Specification Language (CSL)~\cite{speeds:csl}.  The goal of both of these projects was to provide contract-based reasoning for complex embedded systems.  We chose this as our initial pattern language because of the similarity in the contract reasoning approach used by our AGREE tool suite~\cite{CoferNFM2012}.

There is considerable work on real-time specification patterns for different temporal logics.  Konrad and Cheng~\cite{konrad2005real} provide the first systematic study of real-time specification patterns, adapting and extending the patterns of Dwyer~\cite{dwyer1999patterns} for three different temporal logics: TCTL~\cite{alur1991techniques}, MTL~\cite{koymans1990specifying}, and RTGIL~\cite{moser1997graphical}.  Independently, Gruhn~\cite{gruhn2006patterns} constructed a real-time pattern language derived from Dwyer, presenting the patterns as observers in timed automata.  In Konrad and Cheng, multiple (and overlapping) occurrences of patterns are defined in a trace, whereas in Gruhn, only the first occurrence of the pattern considered.  This choice sidesteps the question of adequacy for overlapping triggering events (as discussed in Section~\ref{sec:patterns}), but limits the expressiveness of the specification.  We use a weaker specification language than Konrad~\cite{konrad2005real} which allows better scaling to our analysis, but we also consider multiple occurrences of patterns, unlike Gruhn~\cite{gruhn2006patterns}.  Bellini~\cite{bellini2009expressing} creates a classification scheme for both Gruhn's and Konrad's patterns and provides a rich temporal language called TILCO-X that allows more straightforward expression of many of the real-time patterns.  Like~\cite{konrad2005real}, this work considers multiple overlapping occurrences of trigger events.


The closest work to ours is probably that of Abid et. al~\cite{abid2012real}, who encode a subset of the CSL patterns as observers in a timed extension of Petri nets called TTS, and supplement the observers with properties that involve both safety and liveness in LTL.  For most of the RSL patterns considered, the patterns are only required to hold for the first triggering event, rather than globally across the input trace.  In addition, the use of full LTL makes the analysis more difficult with inductive model checkers.  Other recent work~\cite{etzien2013contracts} considers very expressive real-time contracts with quantification for systems of systems.  This quantification makes the language expressive, but difficult to analyze.

Other researchers including Pike~\cite{Pike2007} and Sorea~\cite{sorea2008modeling} have explored the idea of restricting traces to disallow overlapping events in order to reason about real-time systems using safety properties. The authors of~\cite{Li15} independently developed a similar technique of using a \textit{trigger} variable to specify real-time properties that quantify over events.

\section{Conclusion}

We have presented a method for translating RSL patterns into Lustre
observers. While we only specifically discussed a single pattern in
detail, the techniques we presented can be applied analogously to
other functional or timing patterns. Similarly, the techniques we
presented can be applied to other synchronous data flow languages. The
RSL patterns have been incorporated into the AGREE plugin for the
OSATE AADL integrated development environment. We used these patterns
to show that we could successfully model, and prove properties about,
scheduling constraints for a real-world avionics application. Future
work will focus on lemma generation to improve scalability for
reasoning about real-time properties.



\bibliography{document}
\bibliographystyle{splncs}

\end{document}